\pdfoutput=1
\documentclass[11pt,a4paper]{article}   
\usepackage[utf8]{inputenc}
\usepackage{fourier} 
\usepackage{fullpage}
\usepackage{enumerate}
\usepackage{wrapfig} 
\usepackage{amsmath,amsthm,amssymb}
\usepackage{color}

\usepackage{mathrsfs}
\usepackage[usenames,dvipsnames]{xcolor}
\usepackage[colorlinks=true,linkcolor=Blue,citecolor=Green,urlcolor=Black]{hyperref}
\usepackage{amssymb,bm}

\usepackage{mathtools}
\usepackage[boxed, ruled]{algorithm2e}
\usepackage{framed}
\usepackage{float}
\usepackage{breakcites}

\usepackage[autostyle]{csquotes}
\usepackage[
backend=biber,
style=alphabetic,
maxbibnames=99,
maxalphanames=99,
maxnames=99,
sorting=nyt,
firstinits=true,
url=false,
doi=false,
eprint=true,
isbn=false
]{biblatex}
\usepackage{utf8math}
\usepackage{todonotes}

\DeclareMathAlphabet{\mathcal}{OMS}{cmsy}{m}{n}

\newtheorem*{theorem*}{Theorem}

\newtheorem{theorem}{Theorem}[section]
\newtheorem{lemma}[theorem]{Lemma}

\theoremstyle{definition}

\theoremstyle{remark}
\newtheorem{remark}[theorem]{Remark}

\newcommand{\E}{\mathop{\mathbb E}}

\newcommand{\bd}[1]{\mathrm{bd}\left(#1\right)}

\newcommand{\cvp}[1][]{\ensuremath{\mathrm{CVP}_{#1}}}
\newcommand{\svp}[1][]{\ensuremath{\mathrm{SVP}_{#1}}}
\newcommand{\acvp}[2][(1+\epsilon)]{\ensuremath{#1}-\cvp[#2]}

\renewcommand{\epsilon}{\varepsilon}
\renewcommand{\L}{\mathscr{L}}

\newcommand{\vs}{\mathbf{s}}

\newcommand{\vc}{\mathbf{c}}

\def\R{\mathbb{R}}
\def\Q{\mathbb{Q}}

\def\Z{\mathbb{Z}}

\def\N{\mathbb{N}}
\def\E{\mathbb{E}}

\def\st{\;:\;}

\DeclareMathOperator{\dist}{dist}

\DeclareMathOperator{\vol}{Vol}
\DeclareMathOperator{\conv}{conv}

\author{Thomas Rothvoss\thanks{Supported by NSF CAREER grant 1651861 and a David \& Lucile Packard Foundation Fellowship.} \\
	University of Washington\\
	{\small \texttt{rothvoss@uw.edu}}
	\and
	Moritz Venzin\thanks{Supported by the Swiss National Science Foundation through the project \emph{Lattice Algorithms and Integer Programming} (Nr. 185030) and a Doc.Mobility scholarship within the programm \emph{Integer and Lattice Programming}.} \\
	EPFL\\
	{\small \texttt{moritz.venzin@epfl.ch}}
}
\date{\today}

\title{Approximate \cvp[]  in time $2^{0.802 \, n}$ - now in any norm!} 

\begin{document} 
\maketitle
\begin{abstract}
  \noindent 
We show that a constant factor approximation of the shortest and closest lattice vector problem in any norm can be computed in time $2^{0.802\, n}$. This contrasts the corresponding $2^n$ time, (gap)-SETH based lower bounds for these problems that even apply for small constant approximation.

For both problems, \svp[] and \cvp[], we reduce to the case of the Euclidean norm. A key technical ingredient in that reduction is a twist
of Milman's construction of an $M$-ellipsoid which approximates any symmetric convex body $K$ with an ellipsoid $\mathcal{E}$
so that $2^{\varepsilon n}$ translates of a constant scaling of $\mathcal{E}$ can cover $K$ and vice versa.
\end{abstract}

\section{Introduction}
\begin{figure}[h]\label{cvpsvpexplained}
	\begin{center}
		\includegraphics[width=1.0\textwidth]{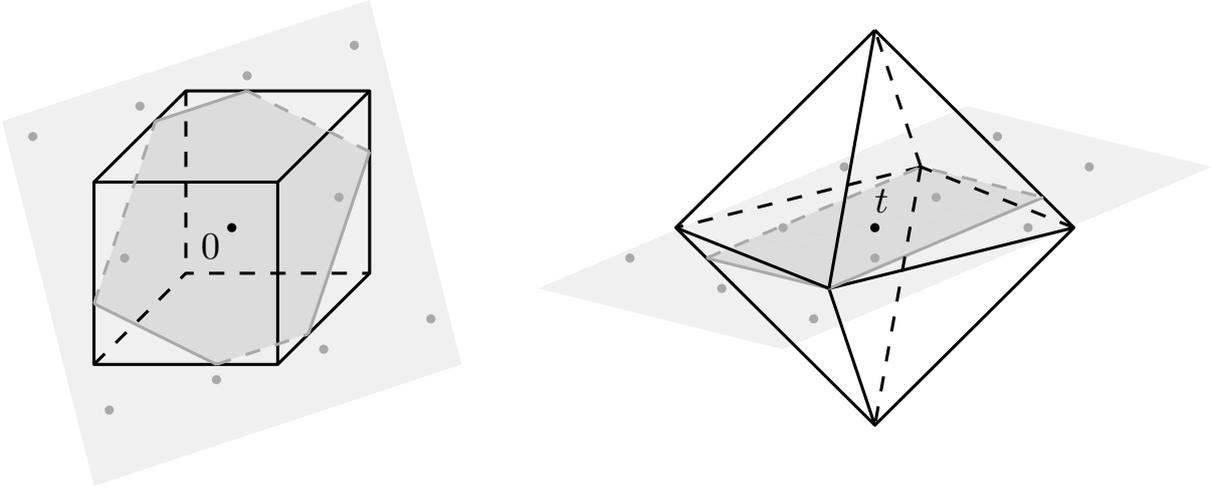}
	\end{center}
	\caption{The shortest and the closest vector problem.} 
\end{figure}
\noindent For some basis $ B \in \R^{d \times n}$, the $d$ dimensional lattice $\L$ of rank $n$ is a discrete subgroup of $\R^d$ given by
\begin{displaymath}
	\L(B) = \{Bx \st x \in \Z^n\}.
\end{displaymath}
\noindent In this work, we will consider the \emph{shortest vector problem} (\svp[]) and the \emph{closest vector problem} (\cvp[]), the two most important computational problems on lattices. Given some lattice, the shortest vector problem is to compute a shortest \emph{non-zero} lattice vector. When in addition some target $t \in \R^d$ is given, the closest vector problem is to compute a lattice vector closest to $t$. \\
Here, "short" and "close" are defined in terms of a given norm $\|\cdot\|_K$, induced by some symmetric convex body $K\subseteq \R^d$ with $\bm{0}$ in its interior. Specifically, $\|x\|_K = \min\{s\in \R_{\geq 0} \mid x \in s\cdot K\}$. When we care to specify what norm we are working with, we denote these problems by $\svp[K]$ and $\cvp[K]$ respectively and by $\svp[p]$ and $\cvp[p]$ respectively for the important case of $\ell_p$ norms. It is important to note that the dimension, rank and span of the lattice have a considerable effect on the norm $\|\cdot\|_K$ that is \emph{induced} on $\L$. For the shortest vector problem, the norm induced on $\L$ corresponds to $K$ \emph{intersected} with the span of $\L$, see Figure \ref{cvpsvpexplained} for an illustration. For different $n$ and different rotations of the lattice, these resulting convex bodies vary considerably. However, when the convex body $K$ is centered in the span of the lattice, these sections of $K$ still satisfy the required properties to define a norm. In particular, up to changing the norm, any ($\alpha$-approximation of the) shortest vector problem in dimension $d \geq n$ can be directly reduced to ($\alpha$-approximate) \svp[] with $d=n$. The situation is slightly more delicate for \cvp[]. Whenever $t \notin \text{span}(\L)$, the function measuring the distance to $t$, i.e. $\|t- \,\cdot\,\|_K$, can be asymmetrical on $\text{span}(\L)$, meaning it does not define a norm on $\text{span}(\L)$. This can be seen by lifting $t$ and the cross-polytope with it in Figure \ref{cvpsvpexplained}. However, up to a loss in the approximation guarantee, we can always take the target to lie in $\text{span}(\L)$ and consider the norm induced by $K$ intersected with $\text{span}(\L)$. More precisely, we can always reduce $(2\cdot\alpha+1)$-approximate \cvp[] to $\alpha$-approximate \cvp[] in any norm with $d = n$. This loss in the approximation factor is not surprising, seeing that exact \cvp[] under general norms is extremely versatile. In fact, \emph{Integer Programming} with $n$ variables and $m$ constraints reduces to $\cvp[\infty]$ on a $m$-dimensional lattice of rank $n$.

%
%
%
%

Both \svp[] and \cvp[] and their respective (approximation) algorithms have found considerable applications. These include Integer Programming~\cite{DBLP:journals/mor/Lenstra83, DBLP:journals/mor/Kannan87}, factoring polynomials over the rationals~\cite{LLL} and cryptanalysis~\cite{Odlyzko90therise}. On the other hand, the security of recent cryptographic schemes are based on the worst-case hardness of (approximations of) these problems~\cite{Ajtai_hardonaverage,Regev_crypto,Gentry_crypto}. In view of their importance, much attention has been devoted to understand the complexity of \svp[] and \cvp[]. In~\cite{vanEmdeBoas81,svp_hard_ajtai,micciancio2001shortest, Arora:1995:PCP:220989,DBLP:journals/combinatorica/DinurKRS03,svp_hard_khot,regev_rosen,svp_hard_regevhaviv}, both \svp[] and \cvp[] were shown to be hard to approximate to within almost polynomial factors under reasonable complexity assumptions. However, the best polynomial-time approximation algorithms only achieve exponential approximation factors~\cite{LLL,Babai1986OnLL,schnorr1987hierarchy}. This huge gap is further highlighted by the fact that these problems are in co-NP and co-AM for small polynomial factors~\cite{Goldreich_Goldwasser,Aharonov_Regev,Peikert08}.

The first algorithm to solve \svp[] and \cvp[] in any norm and even the more general integer programming problem with an exponential running time in the rank only was given by Lenstra~\cite{DBLP:journals/mor/Lenstra83}. Kannan 
\cite{DBLP:journals/mor/Kannan87} improved this to $n^{O(n)}$ time and polynomial space\footnote{For the sake of readability we omit polynomials in the encoding length of the matrix $B$ and the target vector $t$ when stating running times and space requirements.}. To this date, the running time of order $n^{O(n)}$ remains best for algorithms only using polynomial space. It took almost 15 years until Ajtai, Kumar and Sivakumar 
presented a randomized algorithm for $\text{SVP}_2$ with time \emph{and} space 
$2^{O(n)}$ and a  $2^{O(1+1/\epsilon)n}$ time and space algorithm for \acvp{2} 
\cite{DBLP:conf/stoc/AjtaiKS01,DBLP:conf/coco/AjtaiKS02}. Here, $c$-\cvp[] is the problem of finding a lattice vector, whose distance to the target is at most $c$ times the minimal distance. 
Blömer and Naewe~\cite{DBLP:journals/tcs/BlomerN09} extended the randomized sieving algorithm of Ajtai et al. to solve $\svp[p]$ and $\cvp[p]$ respectively in $2^{O(d)}$ time and space and $O(1+1/\epsilon)^{2d}$ time and space respectively. For \cvp[\infty], using a geometric covering technique, Eisenbrand et al.~\cite{DBLP:conf/compgeom/EisenbrandHN11} improved this to $O(\log(1+1/\epsilon))^d$ time. This covering idea was adapted in \cite{cvp_m_and_m} to all (sections of) $\ell_p$ norms. Their algorithm for \acvp{p} requires time $2^{O_p(n)}(1+1/\epsilon)^{n/\min(2,p)}$ and is based on the current state-of-the-art, $2^{O(n)}(1 + 1/\epsilon)^n$ deterministic time CVP solver for general (even asymmetric) norms from Dadush and Kun~\cite{DBLP:journals/toc/DadushK16}. 

Currently, exact and singly-exponential time algorithms for \cvp[] are only known for the $\ell_2$ norm. The first such algorithm was developed by~\cite{DBLP:conf/stoc/MicciancioV10} and is \emph{deterministic}. In fact, this algorithm was also the first to solve \svp[2] in deterministic time (as there is a efficient reduction from \svp[] to \cvp[], \cite{cvp_harderthan_svp}) and has been instrumental to give deterministic algorithms for \svp[] and \acvp{}, \cite{DBLP:conf/focs/DadushPV11, DBLP:journals/toc/DadushK16}. Currently, the fastest exact algorithms for \svp[2] and \cvp[2] run in time and space $2^{n}$ and are based on Discrete Gaussian Sampling~\cite{svp_ADRS,gaussian_sampling,DBLP:conf/soda/AggarwalS18}.

Recently there has been exciting progress in understanding the \emph{fine-grained complexity} of exact and constant approximation algorithms for \svp[] and \cvp[]~\cite{cvp_hard_seth,svp_hard_seth,aggarwal2019fine}.  Under the assumption of the \emph{strong exponential time hypothesis (SETH)} and for  $p\neq 0 \pmod{2}$,  exact \cvp[p] and \svp[\infty] cannot be solved in time $2^{(1-\epsilon)n}$. For a fixed $\epsilon > 0$, the dimension of the lattice can be taken linear in $n$, i.e. $d = O_\epsilon(n)$. Under the assumption of a \emph{gap-version} of the  strong exponential time hypothesis \emph{(gap-SETH)} these lower bounds also hold for the approximate versions of \cvp[p] and \svp[\infty]. More precisely, in our setting these results read as follows. For each $\epsilon > 0$ and for some norm $\|\cdot\|_K$ there exists a constant $\gamma_ {\epsilon}>1$ such that there exists no $2^{(1-\epsilon)n}$ algorithm that computes a $\gamma_{\epsilon}$-approximation of \svp[K] and \cvp[K] (where the corresponding target lies in the span of the lattice).

Until very recently, the fastest approximation algorithms for $\svp[p]$ and $\cvp[p]$ did not match these lower bounds by a large margin, even for large approximation factors~\cite{DBLP:conf/latin/Dadush12,svp_l_infty,svp_mukh_l_p}. The only exception was \svp[2] (where no strong, fine-grained lower bound is known) where a constant factor approximation is possible in time $2^{0.802 n}$ and space $2^{0.401 n}$, see~\cite{svp_mic_voulg,sieving2PS, sieving2LWXZ, AggarwalUrsuVaudenay2019}. Last year, Eisenbrand and Venzin presented a $2^{0.802 n + \epsilon d}$ algorithm for $\svp[p]$ and $\cvp[p]$ for all $\ell_p$ norms~\cite{EisenbrandVenzin}. Their algorithm exploits a specific covering of the Euclidean ($\ell_2$) norm-ball by $\ell_p$ norm-balls and uses the fastest (sieving) algorithm for \svp[2] as a subroutine. This approach was then further extended to yield generic, $2^{\epsilon d}$ time reductions from $\svp[q]$ to \svp[p], $q \geq p \geq 2$ and \cvp[p] to \cvp[q], $q \geq p$, see~\cite{reductions_svp_cvp_davidowitz}. While this improved previous algorithms for \svp[p] and \cvp[p] (and even constant factor approximate \cvp[2]), these techniques \emph{only} apply to the very specific case of $\ell_p$ norms with the added restriction that the dimension $d$ is small, further accentuating the issue rank versus dimension. For any other norm and even for $\ell_p$ norms with, say, $d = \Omega(n\cdot\log(n))$ their approach yields no improvement.
\medskip 

\noindent
In this work, we close this gap. Specifically, for any $d$-dimensional lattice of rank $n$ and for any norm, we show how to solve constant factor approximate $\cvp$ and $\svp$ in time $2^{0.802 n}$. 

\begin{theorem*}
  For any lattice $\L \subseteq \mathbb{R}^d$ of rank $n$, any norm $\|\cdot\|_K$ on $\mathbb{R}^d$ and for each $\epsilon>0$, there exists a constant $\gamma_ {\epsilon}$ such that a $\gamma_ {\epsilon}$-approximate solution to \cvp[K] and \svp[K] can be found in time $2^{(0.802 + \epsilon) n}$ and space $2^{(0.401 + \epsilon)n}$. 
\end{theorem*} 

\noindent We note that the constant $0.802$ in the exponent can be replaced by a slightly smaller number, \cite{kabatiansky1978bounds}. Thus, we indeed get the running time as advertised in the title. \\
For the case of the shortest vector problem, we can significantly generalize this result.

\begin{theorem*}
	For any lattice $\L \subseteq \mathbb{R}^d$ of rank $n$, any norm $\|\cdot\|_K$ on $\mathbb{R}^d$ and for each $\epsilon>0$, there exists a constant $\gamma_ {\epsilon}$ such that there is a $2^{\epsilon n}$ time, randomized reduction from $(\alpha\cdot\gamma_ {\epsilon})$-approximate \svp[K] to an oracle for $\alpha$-approximate \cvp[2] (or even $\alpha$-approximate \cvp[] in any norm).
\end{theorem*}
\noindent 
Our main idea is to cover $K$ by a special class of ellipsoids to obtain the approximate closest vector by using an approximate closest vector algorithm with respect to $\ell_2$. This covering idea draws from \cite{EisenbrandVenzin} and is also similar to the approach of Dadush et al. in \cite{DBLP:conf/focs/DadushPV11, DBLP:journals/toc/DadushK16}. Specifically, for any $\epsilon>0$, one can compute some ellipsoid $\mathcal{E}$ of $K$, so that $K$ can be covered by $2^{\epsilon n}$ translates of $\mathcal{E}$, and, conversely, $\mathcal{E}$ can be covered by $2^{\epsilon n}$ translates of $c_\epsilon\cdot  K$. Here, $c_\epsilon$ is a constant that only depends on $\epsilon$. Such a covering will be sufficient to reduce approximate $\svp[K]$ to $2^{\epsilon n}$ calls to an oracle for (approximate) \cvp[2]. Specifically, using an oracle for $\alpha$-\cvp[2], we will obtain a $\alpha \cdot (2\cdot c_\epsilon)$ approximation to the shortest vector problem with respect to $\|\cdot\|_K$. Similarly, given an oracle for approximate \cvp[Q], we can use these covering ideas twice with $K$ and $Q$ respectively to obtain the desired reduction. These reductions are randomized and use lattice sparsification.
This covering idea can be used to solve constant factor approximate \cvp[K] as well. However, we can no longer assume to only have access to an oracle for the approximate closest vector problem. Instead, we will have to use one very specific property of the currently fastest and randomized approximate algorithm for \svp[2] that was first exploited in \cite{EisenbrandVenzin}.

\section{Approximate \svp[] in $2^{0.802 n}$ time}
\label{sec:coverings}

In this section, we describe our main geometric observation and how it leads to an algorithm for the approximate shortest and closest vector problem respectively. In a first part, we state the main geometric theorem and informally present how it leads to a reduction from approximate $\svp[K]$ to an oracle for approximate \cvp[2]. In the second part, we make this formal using lattice sparsification and some further geometric considerations. Finally, in a third part, we show how to replace the oracle for approximate $\cvp[2]$ by an oracle for $\cvp[]$ under any norm.

\subsection{Covering $K$ with few ellipsoids and vice versa}\label{subsec:svp_idea}

The (approximate) shortest vector problem in the norm $\|\cdot\|_K$ can be rephrased as follows.
\begin{center}
	\textit{Does $K$ contain a lattice point different from $\bm{0}$?}
\end{center}
This follows by guessing the length $\ell$ of the shortest non-zero lattice vector, scaling the lattice by $1/\ell$ and then confirm the right guess by finding this lattice vector in $K$. By guessing we mean to try out all possibilities for $\ell$, this can be limited to a polynomial in the relevant parameters.\\
Imagine now that one can cover $K$ using a collection of (Euclidean) balls $\mathcal{B}_i$ of any radii such that
\begin{displaymath}
	K \subseteq \bigcup_{i=1}^N (x_i + \mathcal{B}_i) \subseteq c\cdot K.
\end{displaymath}
To then find a $c$-approximation to the shortest vector, one could use a solver for \cvp[2] using targets $x_1 \cdots, x_N$. However, this naïve approach is already doomed for $K = B_\infty^n$. To achieve a constant factor approximation, $N=n^{O(n)}$ translates of (any scaling of) the Euclidean ball $B_2^n := \{x \in \R^n \mid \|x\|_2 \leq 1\}$ are required and can only be brought down to $N=2^{O(n)}$ if one is willing to settle for $c = O(\sqrt{n})$~\cite{John1948}. \\
However, if we impose a second condition on these balls (and even relax the above condition), this idea will work. Specifically, if we can cover $K=-K \subseteq \R^n$ by $N$ translates of $c_1 \cdot B_2^n$ \emph{and} $B_2^n$ by $N$ translates of $c_2\cdot K$,
\begin{displaymath}
	K \subseteq \bigcup_{i=1}^N (x_i + c_1 \cdot B_2^n) \quad \text{and} \quad B_2^n \subseteq \bigcup_{i=1}^N(y_i + c_2 \cdot K),
\end{displaymath}
we can solve $\alpha\cdot (c_1 \cdot c_2)$-approximate \svp[K] using (essentially) $N$ calls 
to a solver for $\alpha$-$\cvp[2]$. It turns out that such a covering is always possible for any symmetric and convex $K$, and, in particular, the number of translates $N$ can be made smaller than $2^{\epsilon n}$ for any $\epsilon > 0$. To be precise, the covering is by ellipsoids (affine transformations of balls), but we can always apply a linear transformation to restrict to Euclidean unit norm-balls, i.e. we may take $c_1 = 1$. 
The precise properties of the covering are now stated in the following theorem, where we denote by $N(T,L)$ the \emph{translative covering number} of $T$ by $L$, i.e. the least number of translates of $L$ required to cover $T$.

\begin{theorem}\label{covering:construction}
	For any symmetric and convex body $K \subseteq \R^n$ and for any $\epsilon > 0$, there exists an (invertible) linear transformation $T_\epsilon: \R^n \rightarrow \R^n$ and a constant $c_\epsilon \in \R_{>0}$ such that
	\begin{enumerate}
		\item $N(T_\epsilon(K), B_2^n) \leq 2^{\epsilon n}$ (even $\vol(T_\epsilon (K)+B_2^n) \leq 2^{\epsilon n}\cdot \vol(B_2^n)$) and \label{covering:1}
		\item $N(B_2^n, c_\epsilon\cdot T_\epsilon(K)) \leq 2^{\epsilon n}$ (even $\vol(B_2^n + c_\epsilon \cdot T_\epsilon(K)) \leq 2^{\epsilon n}\cdot\vol(c_\epsilon \cdot T_\epsilon(K)))$.\label{covering:2}
	\end{enumerate}
	The linear transformation $T_\epsilon$ can be computed in (randomized)  $n^{O(\log(n))}$ time (times some polynomial in the encoding length of $K$).
\end{theorem}
      
The volume estimate $\text{Vol}(T_\epsilon (K)+B_2^n) \leq 2^{\epsilon n}\cdot \text{Vol}(B_2^n)$ in $(\ref{covering:1})$ is stronger than $N(T_\epsilon(K),B_2^n)\leq 2^{\epsilon n}$, similar in (\ref{covering:2}). It makes the covering of $T_\epsilon(K)$ by translates of $B_2^n$ constructive. 
Indeed, any point inside $T_\epsilon(K)$ will be covered with probability at least $2^{-\epsilon n}$ if we sample a random point within $T_\epsilon(K) + B_2^n$ and place a copy of $B_2^n$ around it. Repeating this for $O(n\cdot \log(n) \cdot 2^{\epsilon n})$ iterations yields, with high probability, a full covering of $T_\epsilon(K)$ by translates of $B_2^n$. See \cite{Marton:covering} for details. \\
We defer the proof of Theorem \ref{covering:construction} to Section~\ref{sec:ellipsoid} and first discuss how we intend to use it to solve the shortest vector problem in arbitrary norms. To do so, we first fix some notations and do some simplifications. We will denote by $\vs$ a shortest, non-zero lattice vector of the given lattice $\L\subseteq \R^n$ with respect to $\|\cdot\|_K$, the norm under consideration. We will assume that $1-1/n \leq \|\vs\|_K \leq 1$, i.e. $\vs \in K \setminus (1-1/n)\cdot K$. We fix $\epsilon > 0$, and denote by $T_\epsilon$ be the linear transformation that is guaranteed by Theorem~\ref{covering:construction}. Up to replacing $\L$ by $T_\epsilon^{-1}(\L)$ and $K$ by $T_\epsilon^{-1}(K)$, we can also assume that $T_\epsilon = \text{Id}$ ($\|\cdot\|_K = \|T_\epsilon^{-1}(\cdot)\|_{T_\epsilon^{-1}(K)}$). \\
We can now describe how we will use the covering guaranteed by Theorem \ref{covering:construction}. Since $K$ is covered by translates of $B_2^n$, there is some translate that holds $\vs$. We denote it by $t + B_2^n$.
Now, suppose there is a procedure that either returns $\vs$ or generates at least $2^{\epsilon n}+1$ \emph{distinct} lattice vectors lying in $t + \alpha\cdot B_2^n$. For the latter case, while these vectors may all have very large norm with respect to $\|\cdot\|_K$ or may even equal the zero vector, by taking pairwise differences, we are still able to find a $\alpha \cdot (2\cdot c_\epsilon)$-approximation to the shortest vector. Indeed, by property (\ref{covering:2}) of Theorem \ref{covering:construction}, $t + \alpha\cdot B_2^n$ can be covered by fewer than $2^{\epsilon n}$ translates of $(\alpha \cdot c_\epsilon)\cdot K$. Thus, one translate of $(\alpha \cdot c_\epsilon)\cdot K$ must hold \emph{two distinct} lattice vectors. Their pairwise difference is then a $\alpha \cdot((1-1/n)^{-1}\cdot 2\cdot c_\epsilon)$-approximation to the shortest vector $\vs$. This is depicted in Figure \ref{fig:2}.

\begin{figure}[h]\label{covering_by_ellipsoids}
	\begin{center}
		\includegraphics[width=0.7\textwidth]{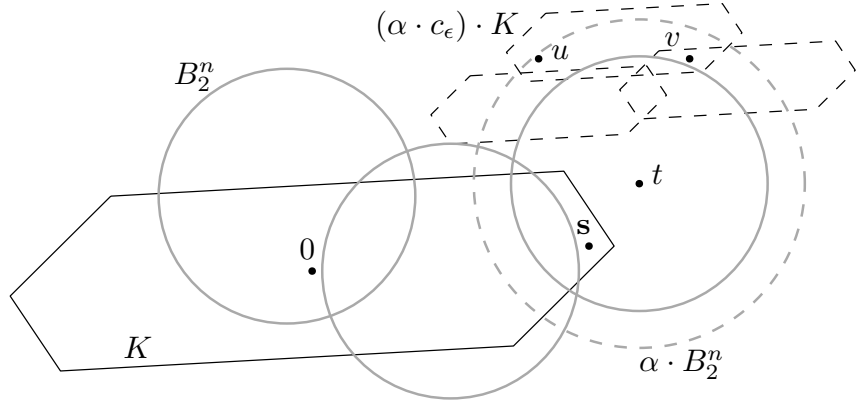}
	\end{center}
	\caption{The covering of $K$ by translates of the Euclidean norm ball. One translate of $(\alpha\cdot c_\epsilon)\cdot K$ covering $t + \alpha\cdot B_2^n$ holds two lattice vectors $u$ and $v$. Their difference $u-v$ is a $(\alpha \cdot \gamma_\epsilon)$-approximation to the shortest vector problem.} \label{fig:2}
\end{figure}
To finish the argument, it remains to argue that such a procedure can be simulated with an oracle for $\alpha$-approximate \cvp[2] at hand. This will be achieved through the use of \emph{lattice sparsification}. This technique will allow us to delete lattice points in an almost uniform manner. Specifically, when the oracle has already returned lattice vectors $v_1, \cdots, v_N$, sparsifying the lattice ensures that with sufficiently high probability, $\vs$ is retained and $v_1, \cdots, v_N$ are deleted. This then \emph{forces} the $\alpha$-approximate \cvp[2] oracle to return a lattice vector $v_{N+1}$ distinct from $v_1, \cdots, v_N$ with $\|t-v_{N+1}\|_K \leq \alpha\cdot \|t-\vs\|_K$. 

\subsection{Approximate \svp[K] using an oracle for approximate \cvp[2]}\label{subsec:svp_formal}
In this subsection we are going to formalize the exponential time reduction from approximate \svp[] to an oracle for (approximate) \cvp[2] as outlined in the previous section. We are going to make use of the following theorem from~\cite{SteDavidowitz16}, slightly rephrased for our purpose.
\begin{theorem}\label{sparsify}
For any prime $p$, any lattice $\L \subseteq \R^d$ of rank $n$ and lattice vectors $w, v_1, v_2, \cdots, v_N \in \L$ with $v_i \notin w + p\cdot\L\setminus\{\bm{0}\}$, one can, in polynomial time, sample a shifted sub-lattice $u + \L'$ with $\L' \subseteq \L$ and $u \in \L$ such that:
\begin{displaymath}
	\Pr\big[w \in u+\L' \text{ and } v_1, \cdots, v_N \notin u+\L'\big] \geq \frac{1}{p} - \frac{N}{p^2}-\frac{N}{p^n}.
\end{displaymath}
\end{theorem}
The condition $v_i \notin w + p\cdot\L\setminus\{\bm{0}\}$ is slightly inconvenient. We will deal with these type of vectors by showing that, for large enough $p$, they will be too large and will not be considered by our $\alpha$-approximate \cvp[2] oracle. This is done by the following lemma. 

\begin{lemma}\label{lemma:length}
	Let $\L$ be a lattice, let $K$ be a symmetric and convex body containing no lattice vector other than $\bm{0}$ in its interior and suppose that $N(K,B_2^n), N(B_2^n,\beta \cdot K)\leq 2^{\epsilon n}$. Then, the following three properties hold:
	\begin{enumerate}
		\item $\forall v \in p\cdot\L\setminus\{\bm{0}\}:$ $\|v\|_2 \geq (2^{-\epsilon n}/\beta)\cdot p $.
		\item $K \subseteq 2^{\epsilon n}\cdot B_2^n$
		\item $(2^{-\epsilon n}/\beta)\cdot B_2^n \subseteq K$
	\end{enumerate}
\end{lemma} 
\begin{proof}
	We first show the last two properties by using the translative covering numbers. Since $N(K,B_2^n)\\\leq 2^{\epsilon n}$, we must have $K \subseteq 2^{\epsilon n}\cdot B_2^n$. To see this, we note that otherwise, the largest segment contained in $K$ cannot be covered using $2^{\epsilon n}$ translates of $B_2^n$. Conversely, since $N(B_2^n,\beta\cdot K)\leq 2^{\epsilon n}$ and the convexity of $K=-K$, the smallest inradius of $K$ cannot be smaller than $2^{-\epsilon n}/\beta$. It follows that $(2^{-\epsilon n}/\beta)\cdot B_2^n\subseteq K$.\\ 
	We can now show the first property. Let $v \in p \cdot \L \setminus \{ \bm{0}\}$. By assumption $\frac{v}{p} \notin \text{int}(K)$ and so $\frac{v}{p} \notin (2^{-\varepsilon n}/\beta)\cdot \textrm{int}(B_2^n)$, where $\textrm{int}(K)$ denotes the interior of the body $K$. This means that $\|v\|_2 \geq (2^{-\varepsilon n}/\beta)\cdot p$.
\end{proof}

We now state our randomized reduction. We rotate the lattice and $K$ so that $\text{span}(\L) = \R^n\times\{0\}^{d-n}$. Up to replacing $K$ by $K\cap\text{span}(\L)$ and deleting the last $d-n$ zeros, we may assume that $d=n$. For details, we refer to the second part of the proof of Lemma \ref{cvp:d=n}. 
We fix some $\epsilon > 0$ and compute the invertible linear transformation $T_\epsilon$ guaranteed by Theorem \ref{covering:construction}. Up to applying the inverse of $T_\epsilon$ and scaling the lattice, we may assume that $T_\epsilon=\text{Id}$ and $1-1/n \leq \|\vs\|_K\leq 1$, where $\vs$ is a shortest lattice vector with respect to the norm $\|\cdot\|_K$. One iteration of the reduction will consist of the following steps and will succeed with probability $2^{-O(\epsilon) n}$.

\begin{enumerate}[(1)]
	\item Fix a prime number $p$ between $2^{3\epsilon n}$ and $2\cdot 2^{3\epsilon n}$.
	\item Sample a random point $t \in K + B_2^n$.
	\item Using the number $p$, sparsify the lattice as in Theorem \ref{sparsify}. Denote the resulting lattice by $u+\L'$.
	\item Run the oracle for $\alpha$-approximate \cvp[2] for $\L'$ with target $t-u$, add $u$ to the vector returned and store it.
	\item Repeat steps (3) and (4) $n^2\cdot 2^{5\epsilon n}$ times. Among the resulting lattice vectors and their pairwise differences, output the shortest (non-zero) with respect to $\|\cdot\|_K$. 
\end{enumerate}

\begin{theorem}\label{thm:reduction_to_CVP_2}
	Let $\L \subseteq \R^d$ be any lattice of rank $n$. For any $\epsilon > 0$, there is a constant $\gamma_\epsilon$, such that there is a randomized, $2^{\epsilon n}$ time reduction from $(\alpha \cdot \gamma_{\epsilon})$-approximate \svp[K] for $\L$ to an oracle for $\alpha$-approximate \cvp[2] for $n$-dimensional lattices.
\end{theorem}

\begin{proof}
	Let us already condition on the event that the sampled point $t$ verifies $\vs \in t + B_2^n$. This happens with probability at least $2^{-\epsilon n}$ by property (\ref{covering:1}) of Theorem \ref{covering:construction}. This probability can be boosted to $1-2^{-n}$ by repeating steps (2) to (5) $O(n\cdot2^{\epsilon n})$ times and outputting the shortest non-zero vector with respect to $\|\cdot\|_K$.\\
	We are now going to show that, with high probability, $\vs$ is going to be retained in the shifted sub-lattice and the lattice vector that is returned in that iteration is different from the ones returned from a previous iteration (where also $\vs$ was retained).
	To this end, denote by $L:=\{v_1, \cdots, v_j\}$ the (possibly empty) list of lattice vectors that were obtained in an iteration where $\vs$ belonged in the corresponding (shifted) sub-lattice. This implies that $\|t-v_i\|_2 \leq \alpha\cdot \|t - \vs\|_2$, and, by the triangle inequality and Lemma~\ref{lemma:length}, $\|v_i\|_2 \leq (2\cdot \alpha)\cdot 2^{\epsilon n}$ for all $i \in \{1, \cdots, j\}$. On the other hand, by slightly rescaling $K$ in Lemma~\ref{lemma:length}, the triangle inequality and our choice of $p$, any vector in $\vs + p\cdot\L\setminus\{\bm{0}\}$ is larger than $2^{2\epsilon n}/c_\epsilon$ in the Euclidean norm. It follows that
	\begin{displaymath}
		v_1, \cdots, v_j \notin \vs + p\cdot\L
	\end{displaymath}
	(We are assuming that $\alpha = 2^{o(n)}$, this is without loss of generality. For $\alpha = 2^{n \log\log(n)/\log(n)}$, Babai's algorithm runs in polynomial time.)\\
	Thus, provided $j \leq 2^{\epsilon n}$ and by Theorem \ref{sparsify}, in any iteration of the algorithm and with probability at least $2^{-3\epsilon n}/2$, we add a lattice vector to the list $L$ that is either distinct from all other vectors in $L$ or that equals $\vs$. In other words, the number of distinct lattice vectors in our list follows a binomial distribution with parameter $2^{-3\epsilon n}/2$. 
        Since we repeat this $n^2\cdot 2^{5\epsilon n}$ times, by Chernov's inequality and with probability at least $1-2^{-n}$, after the final iteration the list $L$ contains at least $2^{\epsilon n}+1$ distinct lattice vectors lying within $t + \alpha\cdot B_2^n$ or contains $\vs$. In the latter case we are done. In the former case, since $N(\alpha\cdot B_2^n, (\alpha\cdot c_\epsilon)\cdot K) \leq 2^{\epsilon n}$, by checking all $2^{10\epsilon n}$ pairwise differences, we will find a non-zero lattice vector with $\| \cdot \|_K$-norm at most $\alpha \cdot(2\cdot c_\epsilon)$. Since $\|\vs\|_K \geq 1-1/n$ and setting $\gamma_\epsilon := (1-1/n)^{-1}\cdot 2\cdot  c_\epsilon$, this vector is a $(\alpha\cdot\gamma_{\epsilon})$-approximation to $\vs$.
\end{proof}

\begin{remark}
	As described, this reduction requires $2^{\epsilon n}$ space. If we assume that the oracle for \cvp[2] is oblivious to previous inquiries, this can be improved to polynomial space. Indeed, in this case, step (5) need only be repeated twice. With  probability $2^{-O(\epsilon) n}$, both lattice vectors lie in the same (scaled) translate of $K$ or equal $\vs$. In our reduction, we can assume that the \cvp[2] oracle is malicious and returns lattice vectors that, dependent on previous inquiries, help us the least.
\end{remark}

\smallskip

The main geometric idea as outlined in the previous subsection leading to Theorem \ref{thm:reduction_to_CVP_2} can be generalized to any pair of norms. Specifically, given norms $\|\cdot\|_K, \|\cdot\|_Q$, we can reduce approximate \svp[K] to an oracle for approximate $\cvp[Q]$. Indeed, up to applying a linear transformation, it follows from Theorem \ref{covering:construction} that, 
\begin{displaymath}
	N(K, Q), N(Q, c_\epsilon \cdot K) \leq 2^{\Omega(\epsilon) n}.
\end{displaymath}
The same covering idea as depicted in Figure \ref{fig:2} (with $B_2^n$ replaced by $Q$) can then be made to work.
\begin{theorem}\label{thm:svp_K_to_cvp_Q}
	Let $\L \subseteq \R^d$ be any lattice of rank $n$ and let $\| \cdot \|_K,\| \cdot \|_Q$ be any norms on $\mathbb{R}^d$. For any $\epsilon > 0$, there is a constant $\gamma_\epsilon$, such that there is a randomized, $2^{\epsilon n}$ time reduction from $(\alpha \cdot \gamma_{\epsilon})$-approximate \svp[K] for $\L$ to an oracle for $\alpha$-approximate \cvp[Q] for $n$-dimensional lattices.
\end{theorem}
\noindent The proof largely follows from the ideas presented in this section and is deferred to the appendix, see subsection \ref{subsec:svp_K_to_cvp_q}.

\section{Approximate \cvp[] in time $2^{0.802 n}$}\label{sec:cvp_in_time_0.802}

In this section we are going to describe a $2^{0.802 n}$ time algorithm for a constant factor approximation to the closest vector problem in any norm. Specifically, for any $d$-dimensional lattice $\L$ of rank $n$, target $t$, any norm $\|\cdot\|_K$ and any $\epsilon > 0$, we are going to show how to approximate the closest vector problem to within a constant factor $\gamma_\epsilon$ in time $2^{(0.802+\epsilon) n }$. The space requirement is of order $2^{(0.401 + \epsilon)n}$.

To achieve this, we will adapt the geometric ideas as outlined in the previous subsection to the setting of the closest vector problem. To do so, we are first going to describe how to restrict to the full-dimensional case.


\begin{lemma}\label{cvp:d=n}
Consider an instance of the closest vector problem, $\cvp[K](\L, t)$, $\L\subseteq \R^d$ of rank $n$. In polynomial time, we can find a lattice $\tilde{\L} \subseteq \R^n$ of dimension $n$, target $\tilde{t}\in \R^n$ and norm $\|\cdot\|_{\tilde{K}}$ so that an $\alpha$-approximation to \cvp[\tilde{K}] on $\tilde{\L}$ with target $\tilde{t}$ can be efficiently transformed in a $(2\alpha+1)$-approximation to $\cvp[K](\L, t)$. Whenever $t \in \text{span}(\L)$, the latter is a $\alpha$-approximation to $\cvp[K](\L, t)$.
\end{lemma}

\begin{proof}
	Let us define $t' := \text{argmin}\{\|t-x\|_K, \, x \in \text{span}(\L)\}$. Such a point may not be unique, for instance for $K = B_1^n$ or $K= B_\infty^n$, but it suffices to consider any point $t'$ realizing this minimum or an approximation thereof. Given a (weak) separation oracle for $K$, this can be computed in polynomial time. We now show that an $\alpha$-approximation to $\cvp[K](\L, t')$ yields a $(2\cdot\alpha+1)$ approximation to $\cvp[K](\L, t)$. Indeed, since $\|t-t'\|_K$ is smaller than $\dist_K(t, \L)$, the distance of $t$ to the closest lattice vector, we have that
	\begin{displaymath}
		\dist_K(t', \L) \leq 2\cdot \dist_K(t, \L).
	\end{displaymath}	
	Denote by $c_\alpha \in \L$ an $\alpha$-approximation to the closest lattice vector to $t'$. By the triangle inequality,
	\begin{displaymath}
		\|t-c_\alpha\|_K \leq \|t-t'\|_K + \|t'-c_\alpha\|_K \leq \dist(t, \L) + \alpha \cdot \dist(t', \L) \leq (2\cdot\alpha +1)\cdot \dist(t, \L).
	\end{displaymath}
	This means that an $\alpha$-approximation to the closest vector to $t'$ is a $(2\cdot\alpha+1)$ approximation to the closest vector to $t$. Now that $t' \in \text{span}(\L)$, we can restrict to the case $d = n$. \\
	Let $\mathcal{O}_n$ be a linear transformation that first applies a rotation sending $\text{span}(\L)$ to $\R^n \times \{0\}^{d-n}$ and then restricts onto its first $n$ coordinates. The transformation $\mathcal{O}_n:\text{span}(\L) \rightarrow \R^n$ is invertible. The $n$-dimensional instance of the closest vector problem is then obtained by setting $\tilde{\L} \leftarrow \mathcal{O}_n(\L)$, $\tilde{t}\leftarrow \mathcal{O}_n(t')$ and $\|\cdot\|_{\tilde{K}}$ where $\tilde{K} \leftarrow \mathcal{O}_n(K)$. Whenever $c_\alpha$ is an $\alpha$-approximation to $\cvp[\tilde{K}](\tilde{\L},\tilde{t})$, 
        $\mathcal{O}_n^{-1}(c_\alpha)$ is a $(2\cdot\alpha+1)$-approximation to $\cvp[K](\L, t)$ (or an $\alpha$-approximation to $\cvp[K](\L, t)$, if $t \in \text{span}(\L)$).
\end{proof}

In our algorithm, we are going to make use of the main subroutine of~\cite{EisenbrandVenzin}. We note that their subroutine is implicit in all sieving algorithms for~\svp[2] and was first described by~\cite{svp_mic_voulg,sieving2PS}.
For convenience, we slightly restate it in the following form. 

\begin{theorem}\label{main_procedure}
	Given $\epsilon > 0$, $R > 0$, $N \in \N$ and a lattice $\L \subseteq \mathbb{R}^d$ of rank $n$, there is a randomized procedure that produces independent samples $v_1, \cdots, v_N \sim \mathcal{D}$, where the distribution $\mathcal{D}$ satisfies the following two properties:
	\begin{enumerate}
		\item Every sample $v \sim \mathcal{D}$ has $v \in \L$ and $\|v\|_2 \leq a_\epsilon \cdot R$, where $a_\epsilon$ is a constant only depending on $\epsilon$.\label{main_procedure:1} 
		\item\label{main_procedure:2} For any $\vs \in \L$ with $\|\vs\|_2 \leq R$, there are distributions $\mathcal{D}^{\vs}_0$ and $\mathcal{D}^{\vs}_1$ and some parameter $\rho_\vs$ with $2^{-\epsilon n} \leq \rho_\vs \leq 1$ such that the distribution $\mathcal{D}$ is equivalent to the following process:
			\begin{enumerate}
				\item \label{main_procedure:2a}With probability $\rho_\vs$, sample $u \sim \mathcal{D}^\vs_0$. Then, flip a fair coin and with probability $1/2$, return $u$, otherwise return $u + \vs$.
				\item With probability $1-\rho_\vs$, sample $u \sim \mathcal{D}^\vs_1$.
			\end{enumerate}
	\end{enumerate}
This procedure takes time $2^{(0.802 + \epsilon)n} + N \cdot 2^{(0.401 + \epsilon)n}$ and requires $N + 2^{(0.401+\epsilon)n}$ space and succeeds with probability at least $1/2$.
\end{theorem} 

With this randomized procedure, we obtain our main result.
\begin{theorem}\label{cvp_to_sieving_2}
	For any $\epsilon > 0$ and lattice $\L$ of rank $n$, we can solve the $\gamma_\epsilon$-approximate closest vector problem for any norm $\| \cdot \|_K$ in (randomized) time $2^{(0.802+\epsilon) n}$ and space $2^{(0.401 + \epsilon)n}$.
\end{theorem}

\begin{proof}
	Using the reduction given by Lemma \ref{cvp:d=n}, we may assume that $\L$ is of full rank, i.e. $\L \subseteq~\R^n$.
	We denote by $\vc$ the closest lattice vector with respect to $\|\cdot\|_K$ to $t$. Up to scaling and applying the inverse of the linear transformation given by Theorem \ref{covering:construction} to $\L$, $K$ and $t$, we may assume that $1-1/n \leq \|\vc-t\|_K \leq 1$, $\text{Vol}(K+B_2^n) \leq 2^{\epsilon n}\cdot\text{Vol}(B_2^n)$ and $N(B_2^n, c_\epsilon\cdot K) \leq 2^{\epsilon n}$ for the constant $c_{\epsilon} >0$.\\
	We now sample a uniformly random point $\tilde{t}$ within $t+K + B_2^n$. With probability at least $2^{-\epsilon n}$, $\vc \in \tilde{t}+B_2^n$. For the remainder of the proof, we condition on $\vc \in \tilde{t}+B_2^n$.  \\
	\begin{figure}[H]\label{CVP-to-sievingA}
	\begin{center}
		\includegraphics[width=5.5cm]{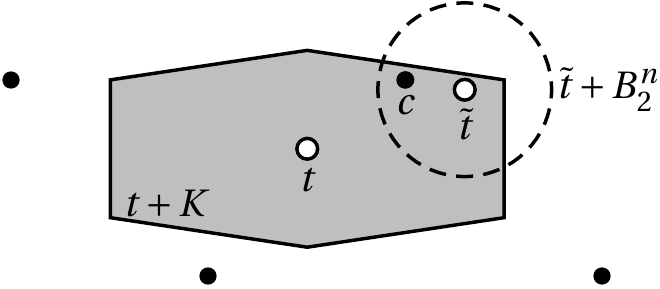}
	\end{center}
	\caption{Finding a new target vector $\tilde{t}$.} 
\end{figure}
	We now use Kannan's embedding technique~\cite{DBLP:journals/mor/Kannan87} and define a new lattice $\L'\subseteq \R^{n+1}$ of rank $n+1$ with the following basis:
	\begin{displaymath}
		\tilde{B} = \begin{pmatrix}
			B & \tilde{t}\\
			0 & 1/n\\
		\end{pmatrix}\in \Q^{(n+1)\times(n+1)}.
              \end{displaymath}
	\begin{figure}[H]\label{CVP-to-sievingB}
	\begin{center}
		\includegraphics[width=11.0cm]{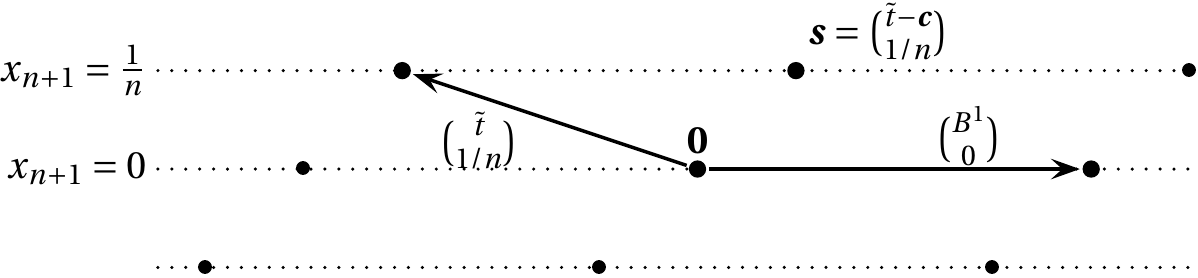}
	\end{center}
	\caption{New lattice $\L'$.} 
\end{figure}
	Finding a ($\alpha$-approximate) closest lattice vector to $\tilde{t}$ in $\L$ is equivalent to finding a ($\alpha$-approximate) shortest lattice vector in $\L' \cap \{x \in \R^{n+1} \mid x_{n+1} = 1/n\}$. The vector $\vs := (\tilde{t}-\vc, 1/n)$ is such a vector (although not necessarily shortest), its Euclidean length is at most $1 + 1/n$. \\
	Now, consider the $n+1$-dimensional scaled Euclidean ball $ ((1+1/n) \cdot a_\epsilon) \cdot B_2^{n+1}$. Here, $a_\epsilon$ is the constant from Theorem \ref{main_procedure}. Each of its $n$-dimensional layers of the form $((1+1/n)\cdot a_\epsilon) \cdot B_2^{n+1} \cap\{x\in \R^{n+1}\mid x_{n+1}=k/n\}$ for $k \in \Z$ can be covered by at most $2^{\epsilon n}$ translates of $((1+1/n)\cdot a_{\epsilon} \cdot c_\epsilon) \cdot K\times\{0\}$. It follows that all lattice vectors inside $a_\epsilon \cdot (1+1/n) \cdot B_2^{n+1} \cap \L'$ can be covered by at most $(2\cdot n\cdot a_\epsilon+1)\cdot 2^{\epsilon n}$ translates of $((1+1/n)\cdot a_{\epsilon} \cdot c_\epsilon) \cdot K\times\{0\}$.\\ 
	We now use the procedure from Theorem~\ref{main_procedure} with $R := 1+1/n$ and sample $N:=2$ lattice vectors from $a_{\varepsilon} \cdot R \cdot B_2^{n+1} \cap \L'$.
        With probability at least $\frac{1}{2} \cdot \rho_{\vs}^2 \geq \frac{1}{2} \cdot 2^{-2\varepsilon n}$ this succeeds and both lattice vectors are generated according to (\ref{main_procedure:2a}). Let us condition on this event.
        That means the sampled lattice vectors are of the form $v_1+\sigma_1\bm{s}$ and $v_2+\sigma_2\bm{s}$
        where $v_1,v_2,\sigma_1,\sigma_2$ are independently distributed with $v_1,v_2 \sim \mathcal{D}_0^{\bm{s}}$ and $\sigma_{1},\sigma_2 \sim \{ 0,1\}$ uniformly. 
        Since $v_1$ and $v_2$ are i.i.d., with probability at least $(2\cdot n\cdot a_\epsilon+1)^{-2}\cdot2^{-2\epsilon n}$, there must be one translate
        of $((1+1/n)\cdot a_{\epsilon} \cdot c_\epsilon)\cdot K \times \{ 0\}$ that contains both $v_1$ and $v_2$. Put differently and using $K = -K$, 
	\begin{displaymath}
		v_1 - v_2 \in (2\cdot (1+1/n)\cdot a_{\epsilon} \cdot c_\epsilon)\cdot K\times\{0\}.
	\end{displaymath}
        Next, we decide the independent coin flips and with probability of 1/4 we have $\sigma_1 = 1$ and $\sigma_2 = 0$.
 We condition on this event. Then the difference of the sampled lattice vectors is
	\begin{displaymath}
		(v_1+\bm{s})-v_2 \in (2\cdot(1+1/n)\cdot a_{\epsilon} \cdot c_\epsilon)\cdot K\times\{0\}+\vs.
	\end{displaymath}
	 We can rewrite it as
	\begin{displaymath}
		(v_1+\bm{s})-v_2 = \begin{pmatrix} u \\ 0 \end{pmatrix} + \vs = \begin{pmatrix}
			u-\vc\\
			0
		\end{pmatrix}+
	\begin{pmatrix}
		\tilde{t}\\
		1/n	
	\end{pmatrix},
	\end{displaymath} 
	where $u \in \L$ and $\|u\|_K \leq (2\cdot(1+1/n)\cdot a_{\epsilon} \cdot c_\epsilon)$. The vector $u-\vc$ can be found by adding $-\tilde{t}$ to the first $n$ coordinates of $(v_1+\bm{s})-v_2$. Then $\vc-u$ will be our approximation to the closest lattice vector. Indeed, by the triangle inequality,
	\begin{displaymath}
		\|t-(\vc-u)\|_K \leq \|t-\vc\|_K + \|u\|_K \leq 1 + (2\cdot(1+1/n)\cdot c_{\epsilon} \cdot a_\epsilon) := \beta_\epsilon.
	\end{displaymath}	
	We set $\gamma_{\epsilon} := (1-1/n)^{-1}\cdot(2\cdot \beta_\epsilon + 1)$; recall that we have used Lemma~\ref{cvp:d=n} and scaled the lattice so that $\|t-\vc\|_K \geq 1-1/n$. 
    The lattice vector $\vc-u$ is a $\gamma_{\epsilon}$-approximation to the closest lattice vector to $t$.\\
	To boost the probability of success from $2^{-\Omega(\epsilon)n}$ to $1-2^{-n}$, we can repeat the steps starting from where we defined $\L'$ $2^{-\Omega(\epsilon)n}$ times and only store the currently closest lattice vector to $t$. Finally, to ensure that a $\tilde{t}$ with $\vc \in \tilde{t} + B_2^n$ is found, we repeat the whole procedure starting from (2) $O(n \cdot 2^{\epsilon n})$ many times. This boosts the overall success probability to $1-2^{-n}$ and yields a total running time of $2^{(0.802 + O(\epsilon))n}$ and space $2^{(0.401+O(\epsilon))n}$. 
\end{proof}

It is unclear whether these ideas can be strengthened to yield an approximation to \cvp[K] using an oracle for approximate \cvp[2] (or even approximate \cvp[Q] for some norm $\|\cdot\|_Q$). Indeed, we crucially rely on the distribution from Theorem \ref{main_procedure} which already solves (and might be more powerful than) approximate \cvp[2]. This procedure is stated in terms of the Euclidean norm, but it does not exploit anything specific about the Euclidean norm other than a bound of $2^{0.401 n}$ on a variant of the kissing number. In particular, this procedure can be made to work for any norm $\|\cdot\|_Q$. In fact, this procedure is inherent in any sieving algorithm for \svp[] and its running time and space requirement only depend on this variant of the kissing number corresponding to $\|\cdot\|_Q$. Specifically, it is the maximum number $N$ of points $v_1, \cdots, v_N\in \bd{Q}$ so that $\|v_i-v_j\|_Q \geq 1-\gamma$ for some $\gamma > 0$. For the case of the Euclidean norm, this can be rephrased in terms of the angular distance. It is then the maximum number $N$ of points $v_1, \cdots, v_N\in \mathbb{S}^{n-1}$ so that $\cos(v_i, v_j) \geq 60^\circ - \tilde{\gamma}$ for some $\tilde{\gamma} > 0$. The approximation guarantee, e.g. $a_\epsilon$ in (\ref{main_procedure:1}) of Theorem \ref{main_procedure}, then also depends on $\gamma$ and $\tilde{\gamma}$ respectively.\\
So, up to some small tweaks to an algorithm for approximate \svp[Q] and with essentially the same running time and space, we obtain the procedure from Theorem \ref{main_procedure} but for $\|\cdot\|_Q$. Using similar covering ideas as for Theorem \ref{thm:svp_K_to_cvp_Q} as outlined in Section \ref{sec:coverings}, we obtain the desired reduction.

\begin{theorem*}[informal]
	For any two norms $\|\cdot\|_K, \|\cdot\|_Q$ and $\epsilon > 0$, we can solve $\gamma_\epsilon$-approximate \cvp[K] using $2^{\epsilon n}$ calls to a sieving algorithm for \svp[Q].
\end{theorem*}

The exact statement and its proof is deferred to the appendix, see Theorem~\ref{cvp_to_sieving_Q}. In Subsection~\ref{subsec:cvp_K_to_sieving_Q}, we also discuss in greater detail what we mean by a sieving algorithm, its relation to this variant of a kissing number and how to adapt the procedure from Theorem \ref{main_procedure} to arbitrary norms.

\section{Constructing the ellipsoid}\label{sec:ellipsoid}

We now discuss Theorem \ref{covering:construction}. Specifically, for a given symmetric convex body $K \subseteq \R^n$ and any fixed $\epsilon > 0$, we are going to outline the construction of a linear transformation $T_\epsilon$ such that
\begin{enumerate}
	\item $\text{Vol}(T_\epsilon(K) + B_2^n) \leq 2^{\epsilon n}\cdot \text{Vol}(B_2^n)$ and 
	\item $\vol(B_2^n + c_\epsilon\cdot T_{\varepsilon}(K)) \leq 2^{\epsilon n}\cdot \vol(c_\epsilon\cdot T_{\varepsilon}(K))$.
\end{enumerate}
The number $c_\epsilon$ only depends on $\epsilon$ and equals $ C^2\cdot \epsilon^{-4}$. Here, $C$ is an universal constant.\\
This construction is based on a procedure called \emph{isomorphic symmetrization}. It is taken almost verbatim from the proof of existence of $M$-ellipsoids due to Milman~\cite{Milman1988IsomorphicSymmetrization}, see also the wonderful textbook of \cite{ArtsteinAvidan2015AsymptoticGA}. In particular, this technique has been made algorithmic by Dadush and Vempala~\cite{DadushVempala13_Mellipsoids} to obtain (deterministic) algorithms for volume computation. Our contribution is thus largely the (simple) observation that this procedure can be stopped earlier to yield the desired properties. For this reason, we prefer to keep this part rather informal and only sketch the relevant techniques\footnote{For the sake of completeness we should point out that the mere existence of such a linear transformation has been known. Pisier~\cite{Pisier1989NewApproachForMilman} proved the existence of a \emph{regular} $M$-ellipsoid. More precisely, he proved that for any constant $0 \leq p < 2$ and any symmetric convex body $K \subseteq \mathbb{R}^n$, there exists an ellipsoid $\mathcal{E} \subseteq \mathbb{R}^n$ so that
	$N(K,t\mathcal{E}), N(\mathcal{E},tK),N(K^{\circ},t\mathcal{E}^{\circ}),N(\mathcal{E}^{\circ},tK^{\circ}) \leq \exp\big(A_pn/t^p\big)$ for all $t \geq 1$. The implication is that this ellipsoid $\mathcal{E}$ approximates $K$ for the \emph{whole} range of $t$ while our
	modification of the isomorphic symmetrization gives an ellipsoid $\mathcal{E}$ that works for a \emph{particular} value of $t$. However, it is less clear how to make Pisier's complex-analytic argument constructive.}.
\medskip

We first introduce the \emph{polar} (or \emph{dual}) $K^\circ$ of $K$. Given $K \subseteq \R^n$, we define
\begin{alignat*}{1}
	K^\circ := \big\{x \in \R^n \mid x^T y \leq 1, \, \forall y \in K\big\}.
\end{alignat*}
Whenever $K = -K \subseteq \R^n$ is full dimensional and with $\bm{0}$ in its interior, so is $K^\circ$. Note that applying a linear invertible transformation $A$ to $K$ transforms $K^\circ$ by $A^{-1}$, i.e. $(A\cdot K)^\circ = A^{-1}\cdot K^\circ$. We can now define the \emph{M-values} $M(K)$ and $M(K^\circ)$ of $K$ and $K^\circ$ respectively. 
\begin{displaymath}
	M(K) = \E_{x \sim \mathbb{S}^{n-1}}[\|x\|_K] \,\,\text{      and      }\,\, M(K^\circ)= \E_{x \sim \mathbb{S}^{n-1}}[\|x\|_{K^\circ}].
\end{displaymath}
Here $\mathbb{S}^{n-1} := \{ x \in \mathbb{R}^n \mid \|x\|_2 = 1\}$ is the sphere.
These quantities can be estimated to arbitrary precision using samples from $K$ and $K^\circ$ respectively.\\
We can now state the celebrated \emph{$M\, M^\circ$ estimate} which follows from combining results of Pisier~\cite{Pisier1980PisiersInequalityOnKConvexity} and of Figiel and Tomczak. For any symmetric convex $K$, there is a linear transformation $T$ and a universal constant $C$ such that
\begin{alignat}{1}
	M(T(K))\cdot M((T(K))^\circ) \leq C\cdot\log(d(K, B_2^n)).\label{def:MMestimate}
\end{alignat}
Here, $d(K, B_2^n)$ is the \emph{Banach-Mazur distance} of $K$ to the Euclidean ball. Specifically, it is the smallest number $s$ such that $B_2^n \subseteq A\cdot K \subseteq s\cdot B_2^n$ for some affine map $A:\R^n \rightarrow \R^n$. It is known that for any symmetric convex body $K \subseteq \mathbb{R}^n$, one always has  $d(K,B_2^n) \leq \sqrt{n}$~\cite{John1948}. \\
The linear transformation $T:\R^n \rightarrow \R^n$ in \eqref{def:MMestimate} can be calculated in (randomized) polynomial time to within arbitrary precision by solving the following \emph{convex} program:
\begin{gather*}
	\max \det (T)\\
	\E_{x\in \gamma_n}\big[\|T(x)\|_K^2\big]^{1/2} \leq 1\\
	T \in \R^{n\times n} \text{ positive-definite}.
\end{gather*}
Here $\gamma_n$ denotes the \emph{Gaussian distribution} on $\R^n$ with density function given by $\frac{1}{(2\pi)^{n/2}}\cdot e^{-\|x\|_2^2/2}$.
In order to justify that the above is indeed a convex program, note that the map $T \mapsto \E_{x \in \gamma_n}[\|T(x)\|_K^2]^{1/2}$ is a \emph{(matrix) norm} and the map $\log \det(T)$ is concave on the cone of positive-definite matrices. 

We can now define an iterative procedure. To initialize it, we set $K_0 \leftarrow K$ and find the linear transformation $T_0$ (guaranteed by the $M\,M^\circ$ estimate) such that
\begin{displaymath}
	M(T_0(K_0))\cdot M((T_0(K_0))^\circ) \leq C\cdot \log(d(K_0, B_2^n)).
\end{displaymath}
We now set $\alpha_0 := \max\{d(K_0, B_2^n)^{1/4}, \epsilon^{-1/2}\}$. Formally, we can replace the number $d(K_0, B_2^n)$ by a factor $2$ approximation which can be guessed or enumerated, we discuss this at the end of the description of the procedure. We can now define the next convex body
\begin{displaymath}
	K_1 = \conv\Big( \big(T_0(K_0) \cap ((M((T_0(K_0))^\circ)\cdot\alpha_0)\cdot B_2^n\big) \cup \frac{1}{M(T_0(K_0))\cdot \alpha_0}\cdot B_2^n\Big).
\end{displaymath}
This new body $K_1$ is contained in the ball of radius $M((T_0(K_0))^\circ)\cdot\alpha_0$ and contains the ball of radius $\frac{1}{M(T_0(K_0))\cdot\alpha_0}$. 
In particular, its Banach-Mazur distance to the Euclidean ball has dropped to
(at most) $M(T_0(K_0))\cdot M(T_0(K_0^\circ))\cdot\alpha_0^{2} \leq (4\cdot C)\cdot \alpha_0^{2}\cdot \log(\alpha_0) = O(d(K_0, B_2^n)^{1/2}\cdot\log(d(K_0, B_2^n)))$. \\
Crucially, we have the following volume estimate. For any symmetric convex body $P \subseteq \R^n$, 
\begin{alignat*}{2}
	2^{-Cn/\alpha_0^2}\cdot\vol(K_1 + P) &\leq \vol(T_0(K_0)+ P) &&\leq 2^{Cn/\alpha_0^2}\cdot \vol(K_1 + P).
\end{alignat*}
The constant in the exponent is universal, so we just denote it by $C$ as well. For a proof, we refer to Chapter~8 of \cite{ArtsteinAvidan2015AsymptoticGA}. \\
We now define $K_j$, for $j \geq 2$. Given $K_{j-1}$, we first compute the linear transformation $T_{j-1}$ given by the $M \, M^\circ$ estimate, and, analogously to before with $T_0(K_0)$ replaced by $T_{j-1}(K_{j-1})$ and with $\alpha_{j-1} = \max\{d(K_{j-1}, B_2^n)^{1/4}, \epsilon^{-1/2}\}$, we define
\begin{alignat*}{1}
	K_j = \conv\Big(\big(T_{j-1}(K_{j-1}) \cap (M(T_{j-1}(K_{j-1})^\circ)\cdot\alpha_{j-1})\cdot B_2^n\big) \cup \frac{1}{M(T_{j-1}(K_{j-1}))\cdot \alpha_{j-1}}\cdot B_2^n\Big).
\end{alignat*}
The volume estimate for $K_j$ and $T_{j-1}(K_{j-1})$ are as follows. For any symmetric convex body $P \subseteq \R^n$, 
\begin{gather}\label{volume_estimates_Kj}
	2^{-Cn/\alpha_{j-1}^2}\cdot\vol(K_j + P) \leq \vol(T_{j-1}(K_{j-1})+ P) \leq 2^{Cn/\alpha_{j-1}^2}\cdot \vol(K_j + P)
\end{gather}
In each iteration $K_{j-1} \rightarrow K_j$, the respective Banach Mazur distance to the Euclidean ball is dropping. Specifically,
if $d(K_{j-1}, B_2^n) \geq \Omega(\epsilon^{-1/2})$ and for $\varepsilon$ small enough, we have
\begin{displaymath}
	d(K_j, B_2^n) \leq (4\cdot C)\cdot d(K_{j-1}, B_2^n)^{1/2}\cdot \log(d(K_{j-1}, B_2^n)) \leq \frac{1}{2}d(K_{j-1}, B_2^n).
\end{displaymath}
Thus, for some iteration $t$, we will obtain a convex body $K_t$ with
\begin{alignat}{1}\label{roundness}
	\tfrac{\varepsilon}{C} \cdot B_2^n \subseteq K_t \subseteq \tfrac{C}{\epsilon}\cdot B_2^n.
\end{alignat}
These inclusions are achieved by scaling $K_t$ by (approximately) $M(T_{t-1}(K_{t-1}))$.\\
In particular, we have $\alpha_{t-1} \geq \epsilon^{-1/2}$. Since $\frac{1}{\alpha_0^2} + \cdots + \frac{1}{\alpha_{t-1}^2} \leq \epsilon\cdot O(1)$, we can iteratively combine the volume estimates in ($\ref{volume_estimates_Kj}$) to arrive at
\begin{gather}
	2^{-O(\epsilon)  n}\cdot\vol(K_t + P) \leq \vol(T_{t-1} \cdots T_0(K)+ P) \leq 2^{O(\epsilon) n}\cdot \vol(K_t + P).\label{volume_inequalities}
\end{gather} 
For $P = (C\cdot \epsilon^{-2})\cdot B_2^n$, first using the rightmost inequality in (\ref{volume_inequalities}) and then the rightmost inclusion in (\ref{roundness}),
\begin{gather}\label{conclusion_1}
	\vol(T_{t-1} \cdots T_0(K)+ (C\cdot\epsilon^{-2})\cdot B_2^n) \leq 2^{O(\epsilon)n}\cdot \vol(K_t + (C\cdot\epsilon^{-2})\cdot B_2^n) \leq 2^{O(\epsilon)n}\cdot \vol((C\cdot\epsilon^{-2})\cdot B_2^n).
\end{gather}
On the other hand, setting $P = (\epsilon^2\cdot C^{-1})\cdot B_2^n$ and using the rightmost inequality in (\ref{volume_inequalities}) and then the leftmost inclusion in (\ref{roundness}) combined with the leftmost inequality in (\ref{volume_inequalities}),
\begin{gather}\label{conclusion_2}
	\vol(T_{t-1}\cdots T_0(K) + (\epsilon^2\cdot C^{-1})\cdot B_2^n) \leq 2^{\Omega(\epsilon)n}\cdot \vol(K_t + (\epsilon^2\cdot C^{-1})\cdot B_2^n) \leq 2^{\Omega(\epsilon)n}\cdot \vol(T_{t-1}\cdots T_0(K)).
\end{gather}
We set $T_\epsilon := C^{-1}\cdot \epsilon^{2}\cdot T_{t-1}\cdots T_0$ and $c_\epsilon := C^2\cdot \epsilon^{-4}$. By (\ref{conclusion_1}) and (\ref{conclusion_2}), 
\begin{displaymath}
\text{Vol}(T_\epsilon(K) + B_2^n) \leq 2^{\epsilon n}\cdot \text{Vol}(B_2^n)
\end{displaymath}
and 
\begin{displaymath}
	\vol(B_2^n + c_\epsilon\cdot T_{\varepsilon}(K)) \leq 2^{\epsilon n}\cdot \vol(c_\epsilon\cdot T_{\varepsilon}(K))
\end{displaymath}
immediately follow (up to replacing $\epsilon$ ($C$) by a fraction (multiple) of itself in the beginning). This concludes the description of the procedure.\\
We now discuss two implementation details. First, we have relied on guessing $d(K_i, B_2^n)$ up to a factor of $2$ and seem to know the right iteration $t$ at which to stop. This is without loss of generality. Indeed, since $1 \leq d(K_i, B_2^n) \leq \sqrt{n}$ and $d(K_i, B_2^n) \leq \frac{1}{2}\cdot d(K_i, B_2^n)$, there are at most $\log(n)^{\log(n)}$ possibilities in total. We can either guess and succeed with probability $\log(n)^{-\log(n)}$, or we return $\log(n)^{\log(n)}$ different convex bodies, one of which verifies the desired properties. Both options are fine for our purpose.\footnote{Alternatively, the proof could also be modified to use an overestimate on $d(K_i,B_2^n)$ instead. We refrain from this to keep the proof readable.}
Second, we note that we assume the existence of a (weak) separation oracle for the original body $K_0 := K$. This is sufficient to construct a separation oracle for the intermediate bodies $K_j$ and to compute their respective linear transformations guaranteed by the $M\,M^\circ$ estimate in quasi-polynomial time. Indeed, each such body is of the form $K_j = \text{conv}( (\tilde{K}_{j-1} \cap R_j B_2^n) \cup r_jB_2^n)$ where $\tilde{K}_{j-1} := T_{j-1}(K_{j-1})$ and $0 < r_j \leq R_j$ are the chosen radii. Using the ellipsoid method and the equivalence of optimization and separation, we can solve the separation problem for $K_{j}$, and thus compute the linear transformation, using polynomially many calls to a separation oracle for
$K_{j-1}$. Since there are at most $O(\log(n))$ iterations to consider, each call to a separation oracle for $K_j$ can be evaluated by $n^{O(\log(n))}$ calls to the separation oracle for $K_0$ and results in an overall running time of $n^{O(\log(n))}$. For details, we refer to \cite{GroetschelLovaszSchrivjer88,DadushVempala13_Mellipsoids}. 
\newpage
\printbibliography

\newpage

\section{Appendix}
\subsection{Proof of Theorem \ref{thm:reduction_to_CVP_2}: Reducing \svp[K] to an oracle for approximate \cvp[Q]}\label{subsec:svp_K_to_cvp_q}

We can now show that for any two norms $\|\cdot\|_K, \|\cdot\|_Q$, approximate $\svp[K]$ reduces to (an oracle for) approximate $\cvp[Q]$. 
To this end, set $d=n$ (by intersecting $K$ with $\text{span}(\L)$ and rotating if necessary) and fix $\epsilon > 0$. Denote by $T^{K}_\epsilon(\cdot), T^{Q}_\epsilon(\cdot)$ the respective linear transformations guaranteed by Theorem~\ref{covering:construction}. Up to replacing $K$ by $T_\epsilon^K(K)$ and $Q$ by $T_\epsilon^Q(Q)$, we may assume that $T_\epsilon^K = T_\epsilon^Q = \text{Id}$. Indeed, an oracle for approximate \cvp[Q] is equivalent to an oracle for \cvp[T(Q)] for $T$ some invertible linear transformation. E.g. $T^{-1}(v)$ is a solution to $\alpha$-$\cvp[Q]$ (with lattice $\L$ and target $t$), where $v$ is a solution to $\alpha$-\cvp[Q] for the lattice $T(\L)$ and target $T(t)$. Similarly for \svp[K]. We can thus assume that
\begin{displaymath}
	N(K,c_{\varepsilon}\cdot Q), N(Q,c_{\varepsilon}\cdot K) \leq 2^{2\epsilon n},
\end{displaymath}
and that the corresponding volume estimates hold. 

\begin{lemma} \label{lem:covering_construction_K_Q}
	Let $T_{\varepsilon}^K : \mathbb{R}^n \to \mathbb{R}^n$ be the linear transformation satisfying Lemma~\ref{covering:construction}. Then for any symmetric convex bodies
	$K,Q \subseteq \mathbb{R}^n$ the positions $\tilde{K} := T_{\varepsilon}^K(K)$ and $\tilde{Q} := T_{\varepsilon}^Q(Q)$ satisfy
	\begin{enumerate}
		\item $N(\tilde{K},c_{\varepsilon}\cdot\tilde{Q}) \leq 2^{2\varepsilon n}$ and $\vol(\tilde{K} + c_{\varepsilon}\cdot\tilde{Q}) \leq 2^{2\varepsilon n} \cdot \vol(c_{\varepsilon}\cdot\tilde{Q})$. \label{item:covering_construction_K_Q_I}
		\item $N(\tilde{Q},c_{\varepsilon}\cdot\tilde{K}) \leq 2^{2\varepsilon n}$ and $\vol(\tilde{Q} + c_{\varepsilon}\cdot\tilde{K}) \leq 2^{2\varepsilon n}\cdot\vol(c_{\varepsilon}\cdot\tilde{K})$. \label{item:covering_construction_K_Q_II}
	\end{enumerate}
\end{lemma}
\begin{proof}
	We use the properties provided by Lemma~\ref{covering:construction} to bound
	\begin{eqnarray*}
		\vol(\tilde{K} + c_{\varepsilon}\cdot\tilde{Q}) &\leq& N(\tilde{K}+c_{\varepsilon}\cdot\tilde{Q}, B_2^n+c_{\varepsilon}\cdot\tilde{Q}) \cdot \vol(B_2^n+c_{\varepsilon}\cdot\tilde{Q}) \\
		&\stackrel{(*)}{\leq}& \underbrace{N(\tilde{K},B_2^n)}_{\leq 2^{\varepsilon n}} \cdot \underbrace{\vol(B_2^n +c_{\varepsilon}\cdot\tilde{Q})}_{\leq 2^{\varepsilon n}\cdot \vol(c_{\varepsilon}\cdot\tilde{Q})} \leq 2^{2\varepsilon n}\cdot \vol(c_{\varepsilon}\cdot \tilde{Q})
	\end{eqnarray*}
	In $(*)$ we use the fact that for any symmetric convex bodies $A,B,C$ one has $N(A+C,B+C) \leq N(A,B)$ as any covering of $A$ with translates of $B$ implies a covering of $A+C$ with translates of $B+C$. As discussed earlier, this in particular implies the
	covering estimate of $N(\tilde{K},c_{\varepsilon}\cdot\tilde{Q}) \leq 2^{2\varepsilon n}$. Direction~\ref{item:covering_construction_K_Q_II} is analogous; simply switch the roles
	of $\tilde{K}$ and $\tilde{Q}$.
\end{proof}

Finally, by scaling the lattice, we may assume that $1-1/n \leq \|\vs\|_K \leq 1$, where $\vs$ is the shortest lattice vector with respect to $\|\cdot\|_K$. We can now describe one iteration of reduction that succeeds with probability at least $2^{-\Omega(\epsilon)n}$.

\begin{enumerate}[(1)]
	\item Fix a prime number $p$ between $2^{5\epsilon n}$ and $2\cdot 2^{5\epsilon n}$.
	\item Sample a random point $t \in K + c_\epsilon \cdot Q$. 
	\item Using the number $p$, sparsify the lattice as in Theorem \ref{sparsify}. Denote the resulting lattice by $u+\L'$.
	\item Run the oracle for $\alpha$-approximate \cvp[Q] for $\L'$ with target $t-u$, add $u$ to the vector returned and store it.
	\item Repeat steps (3) and (4) $2^{8\epsilon n}$ times. Among the resulting lattice vectors and their pairwise differences, output the shortest (non-zero) with respect to $\|\cdot\|_K$. 
\end{enumerate}

\begin{theorem*}
	Let $\L \subseteq \R^d$ be any lattice of rank $n$ and let $\| \cdot \|_K,\| \cdot \|_Q$ be any norms on $\mathbb{R}^d$. For any $\epsilon > 0$, there is a constant $\gamma_\epsilon$, such that there is a randomized, $2^{\epsilon n}$ time reduction from $(\alpha \cdot \gamma_{\epsilon})$-approximate \svp[K] for $\L$ to an oracle for $\alpha$-approximate \cvp[Q] for $n$-dimensional lattices.
\end{theorem*}

\begin{proof}
	We already condition on the event that $\vs \in t + c_\epsilon\cdot Q$, this holds with probability at least $2^{-2\epsilon n}$. 
	Since Lemma \ref{covering:construction} holds for both $K$ and $Q$, we can use Lemma \ref{lemma:length} twice with $K$ and $B_2^n$ and $Q$ and $B_2^n$ respectively to find
	\begin{displaymath}
		(2^{-2\epsilon n}/c_\epsilon)\cdot Q \subseteq K \subseteq (2^{2\epsilon n}\cdot c_\epsilon)\cdot Q.
	\end{displaymath}
	Since $t \in K + B_2^n + Q$, it follows that $\|t\|_Q \leq 3\cdot c_\epsilon \cdot 2^{2\epsilon n}$. By the triangle inequality, any lattice vector $v \in \vs + p\cdot \L$ must have $\|t-v\|_Q > 2^{2\epsilon n}$. So whenever $\vs$ is retained in the sparsified lattice in the third step, the oracle must return a lattice vector $w\notin \vs + p\cdot \L$ since $\|t-w\|_Q \leq (\alpha \cdot c_\epsilon) \cdot \|t-\vs\|_Q \leq \alpha\cdot c_\epsilon \ll 2^{2\epsilon n}$.\\
	We can now proceed similarly to the proof of Theorem \ref{thm:reduction_to_CVP_2}. Denote by $L := \{v_1, \cdots, v_j\}$ the list of lattice vectors that were obtained in an iteration where $\vs$ belonged to the (shifted) sub-lattice. By Theorem~\ref{sparsify}, whenever $j \leq 2^{\epsilon n}$ and with probability at least $2^{-5\epsilon n}/2$, in each iteration a lattice vector is added to the list $L$ that is distinct from all lattice vectors in $L$ or that equals $\vs$. By Chernov's inequality and with probability at least $1-2^{-n}$, after the final iteration, the list $L$ contains $2^{2\epsilon n} + 1$ distinct lattice vectors contained inside $t + c_\epsilon \cdot K$ or contains $\vs$. In the latter case, we are done. In the former case, by \ref{item:covering_construction_K_Q_II}, $(\alpha \cdot c_\epsilon)\cdot Q$ can be covered by fewer than $2^{2\epsilon n}$ translates of $(\alpha\cdot c_\epsilon^2)\cdot Q$, i.e. $N((\alpha\cdot c_\epsilon)\cdot Q, (\alpha\cdot c_\epsilon^2) \cdot K) \leq 2^{2\epsilon n}$.
	Thus, assuming we have at least $2^{2\epsilon n} +1$ such lattice vectors, there must be two distinct lattice vectors, $v_1$ and $v_2$ say, that are contained inside the same translate of $(\alpha\cdot c_\epsilon^2) \cdot K$. By the symmetry of $K$, 
	\begin{displaymath}
		v_1 - v_2 \in (2\cdot \alpha \cdot c_\epsilon^2)\cdot K \setminus \{\bm{0}\}.
	\end{displaymath}
	Their difference $v_1 - v_2$ is then the desired $(\alpha\cdot\gamma_{\epsilon})$-approximation to $\vs$ (with $\gamma_\epsilon := 2\cdot (1-1/n)^{-1}\cdot c_\epsilon^2$).
	Repeating steps (1) to (5) $2^{3\epsilon n}$ times boosts the overall probability of success to $1-2^{-n}$.
	
\end{proof}

\subsection{Approximate \cvp[K] through a sieving algorithm for \svp[Q]}\label{subsec:cvp_K_to_sieving_Q}

In Section~\ref{sec:cvp_in_time_0.802}, we have seen how approximate $\cvp[K]$ can be solved using some technical property inherent in any sieving algorithm based on~\cite{DBLP:conf/stoc/AjtaiKS01} for approximate \svp[2]. More specifically, we have assumed that we are given access to a certain procedure that samples independent lattice vectors according to some special distribution, see Theorem \ref{main_procedure}. The procedure and its time and space requirements depend on (a slightly stronger variant of) the kissing number with respect to the Euclidean norm that equals $2^{0.401 n}$. While this procedure could be made to work in any norm, only the Euclidean norm was considered because among all \emph{known} bounds on kissing numbers, the one for the Euclidean norm is lowest. Here, we show that this is the only reason. More specifically, we show that if for \emph{some} norm $\|\cdot\|_Q$ the corresponding kissing number (i.e. the variant relevant for sieving) is better than the one for the Euclidean norm, then this would directly improve the running time for \cvp[] in \emph{any} norm. 
\smallskip

Before we state the corresponding generalisation of the procedure Theorem \ref{main_procedure} to norms other than the Euclidean norm, we discuss how this procedure relates to a variant of the kissing number. Given some convex body $L = -L\subseteq \R^n$, the kissing number for $L$ is the maximum number of non-overlapping translates of $\tfrac{1}{2}\cdot L$ with centers on $\bd{L}$. Equivalently, it is the maximum number of points $v_1, \cdots, v_N \in \bd{L}$ so that $\|v_i-v_j\|_L \geq 1$ for all $i, j \in [N]$, $i \neq j$. It is now easy to see how this relates to sieving algorithms as first described by~\cite{DBLP:conf/stoc/AjtaiKS01}. In sieving, we sample exponentially many lattice vectors and take pairwise differences to obtain shorter and shorter lattice vectors with respect to $\|\cdot\|_L$. Clearly, the kissing number with respect to $L$ is a lower bound on the number of lattice vectors that we need to sample. Indeed, we need to sample at least this many lattice vectors just to guarantee that the difference of two of them becomes strictly smaller in order to find short vectors. The space requirement of this procedure would equal the kissing number and the time requirement the square of this number, since we need to check all pairs. For some special norms, better data-structures or trade-offs time/space are known,~\cite{svp_l_infty,svp_mukh_l_p}. 
However, it is currently not known and an interesting open question, see~\cite{Talk_SteDavidowitz2020}, whether a bound on the kissing number is sufficient to obtain a sieving algorithm. Indeed, the procedure might only produce short vectors that equal $\bm{0}$ and a technical argument is needed to ensure that this does not happen. For this, a slight variation (strengthening) of the kissing number is needed. Specifically, for $L = -L \subseteq \R^n$ and $\gamma > 0$, we define $\tilde{k}(L, \gamma)$ as the maximum number of points $v_1, \cdots, v_N$ on $L\setminus(1-\gamma/n)\cdot L$ so that $\|v_i- v_j\|_L \geq 1-\gamma$ for all $i \neq j \in [N]$. For the Euclidean norm, i.e. $L = B_2^n$, this condition is equivalent to require that the angle between the $v_i$ and $v_j$ are larger than $60^\circ - \tilde{\gamma}$, where $\tilde{\gamma}$ is a function of $\gamma$. For approximation factors of order $O(1 + 1/\gamma)$, the running time of $\svp[L]$, $\cvp[L]$ and that of procedure in Theorem \ref{main_procedure} are then of order $\tilde{k}(L, \gamma)^{2}$. Whenever special data structures or trade-offs are available for finding close-by vectors with respect to $\|\cdot\|_Q$, this would directly improve the running time. This also motivates finding some norm $\|\cdot\|_Q$ where such data-structures or trade-off are available, even if its (variant of the) kissing number is worse than that of the Euclidean norm, see~\cite{MillerDavidowitz2019} for a recent attempt.\\
We can now fix some origin-symmetric convex body $Q \subseteq \R^n$, and discuss the relevant properties of a procedure similar to Theorem \ref{main_procedure} adapted to the norm $\|\cdot\|_Q$. Since we are planning to use this procedure on lattices on dimension $n+1$ to solve $\cvp[]$ through Kannan's embedding technique, we need to consider a one-dimension-higher version of $Q$. Specifically, $Q^{+1} \subseteq \R^{n+1}$ is defined as follows,
\begin{displaymath}
	Q^{+1} := \{(x, x_{n+1}) \in \R^{n}\times \R \mid x \in Q, x_{n+1} \in [-1, 1]\}.
\end{displaymath}
The convex body $Q^{+1}$ is the $n+1$ dimensional cylinder obtained by taking $Q$ as its base. The purpose of $Q^{+1}$ is that the norm $\|\cdot\|_{Q^{+1}}$ \emph{extends} the norm $\|\cdot\|_Q$ to lattices of dimension $n+1$ while the respective numbers $\tilde{k}(\cdot,\cdot)$ remain comparable. 
\begin{lemma}\label{norm_plus_1}
	Let $Q \subseteq \R^n$ be a origin symmetric convex body. Define $Q \subseteq \R^{+1}$ as above. Then
	\begin{displaymath}
		\tilde{k}(Q^{+1}, \gamma) \leq O(\tfrac{1}{(1-\gamma)}\cdot \tilde{k}(Q, \gamma)).
	\end{displaymath}
\end{lemma}
\begin{proof}
	Consider the maximum number of disjoint translates $v_i + \tfrac{1-\gamma}{2}\cdot Q^{+1}$, $i \in \{1, \cdots, \tilde{k}(Q^{+1}, \gamma)\}$, where the $v_i$ lie in $Q^{+1}\setminus (1-\gamma/n)\cdot Q^{+1}$. Each such translate must intersect at least one slice of the form $Q^{+1} \cap \{x \in \R^{n+1} \mid x_{n+1} = k \cdot (1-\gamma)\}$, for $k \in \{-(1-\gamma)^{-1}, \cdots, (1-\gamma)^{-1}\}$. On the other hand, by the construction of $Q^{+1}$, each such slice can intersect at most $\tilde{k}(Q, \gamma)$ translates of $\tfrac{1-\gamma}{2}\cdot Q^{+1}$. The bound follows.
\end{proof}
We can now adapt the procedure from Theorem \ref{main_procedure} to arbitrary norms. The formal proofs of these statements are easily obtained by adapting the proof in~\cite{sieving2PS} as outlined in the informal discussion in~\cite{EisenbrandVenzin}, e.g. replacing the Euclidean norm $\|\cdot\|_2$ by $\|\cdot\|_{Q^{+1}}$ and using Lemma \ref{norm_plus_1}. 

\begin{theorem}\label{main_procedure_Q}
	Let $Q=-Q\subseteq \R^n$ and define $Q^{+1} = -Q^{+1}\subseteq \R^{n+1}$ as above. 
	Given $\gamma, \epsilon > 0$, $R > 0$, $N \in \N$ and a lattice $\L \subseteq \mathbb{R}^{n+1}$ of rank $n+1$, there is a randomized procedure that produces independent samples $v_1, \cdots, v_N \sim \mathcal{D}$, where the distribution $\mathcal{D}$ satisfies the following two properties:
	\begin{enumerate}
		\item Every sample $v \sim \mathcal{D}$ has $v \in \L$ and $\|v\|_{Q^{+1}} \leq a_{\epsilon, \gamma} \cdot R$, where $a_{\epsilon, \gamma}$ only depends on $\epsilon$ and $\gamma$.\label{main_procedure_Q:1} 
		\item\label{main_procedure_Q:2} For any $\vs \in \L$ with $\|\vs\|_{Q^{+1}} \leq R$, there are distributions $\mathcal{D}^{\vs}_0$ and $\mathcal{D}^{\vs}_1$ and some parameter $\rho_\vs$ with $2^{-\epsilon n} \leq \rho_\vs \leq 1$ such that the distribution $\mathcal{D}$ is equivalent to the following process:
		\begin{enumerate}
			\item \label{main_procedure_Q:2a}With probability $\rho_\vs$, sample $u \sim \mathcal{D}^\vs_0$. Then, flip a fair coin and with probability $1/2$, return $u$, otherwise return $u + \vs$.
			\item With probability $1-\rho_\vs$, sample $u \sim \mathcal{D}^\vs_1$.
		\end{enumerate}
	\end{enumerate}
	This procedure takes time $ \tilde{k}(Q^{+1},\gamma)^{2} \cdot 2^{\epsilon n} + N\cdot \tilde{k}(Q^{+1}, \gamma)\cdot 2^{\epsilon n}$ and $N + \tilde{k}(Q^{+1}, \gamma)\cdot 2^{\epsilon n}$ space and succeeds with probability at least $1/2$.
\end{theorem}
We now how to solve approximate $\cvp[K]$ on $n$-dimensional lattices assuming we have access to the above distribution for some convex body $Q \subseteq \R^n$.

\begin{theorem}\label{cvp_to_sieving_Q}
	Let $Q=-Q \subseteq \R^n$ be such that $\tilde{k}(Q, \gamma) \leq 2^{\beta n}$ for some $\gamma > 0$. Then, for any norm $\|\cdot\|_K$ and in time $2^{(2\beta+\epsilon )n}$ and space $2^{(\beta+\epsilon )n}$, we obtain a $O_{\epsilon, \gamma}(1)$-approximation to \cvp[K] on any lattice of rank $n$.
\end{theorem} 

\begin{proof}
	We first invoke Lemma \ref{cvp:d=n} to restrict to a lattice of dimension $n$. Analogously to subsection \ref{subsec:svp_K_to_cvp_q}, we may assume that Lemma \ref{lem:covering_construction_K_Q} holds with $T_\epsilon^Q = T_\epsilon^K=\text{Id}$. Indeed, it is enough to solve $\cvp[T^{-1}_\epsilon(K)]$ and for any invertible transformation $T$, we have that $\tilde{k}(Q, \gamma) = \tilde{k}(T(Q), \gamma)$. Finally, we scale the lattice so that distance of the closest lattice lattice vector $\vc \in \L$ to $t$ is close to $1$, i.e. $1-1/n \leq \|t-\vc\|_K \leq 1$.\\
	We now sample a random point $\tilde{t}$ within $t + K + c_\epsilon \cdot Q$. By Lemma \ref{lem:covering_construction_K_Q} and with probability at least $2^{-2\epsilon n}$, $\vc \in \tilde{t}+c_\epsilon \cdot Q$, we condition on this event. For this vector $\tilde{t}$ we define the lattice $\L'\subseteq \R^{n+1}$ of rank $n+1$ with the following basis:
	\begin{displaymath}
		\tilde{B} = \begin{pmatrix}
			B & \tilde{t}\\
			0 & 1/n\\
		\end{pmatrix}\in \Q^{(n+1)\times(n+1)}.
	\end{displaymath}
	We now consider the norm $\|\cdot\|_{Q^{+1}}$ induced by $Q^{+1} \subseteq \R^{n+1}$. We define $\vs := (\tilde{t}-\vc, 1/n) \in \R^{n+1}$, observe that $1-1/n \leq \|\vs\|_{Q^{+1}} \leq 1$. By Lemma \ref{lem:covering_construction_K_Q}, each layer of the form $(c_\epsilon \cdot a_{\epsilon, \gamma})\cdot Q^{+1} \cap \{x_{n+1} = 1/n\}$ can be covered using $2^{2\epsilon n}$ translates of $(c_\epsilon^2 \cdot a_{\epsilon, \gamma})\cdot K\times\{0\}$. In particular, all lattice vectors of length at most $a_{\epsilon, \gamma}\cdot c_\epsilon$ (with respect to $\|\cdot\|_{Q^{+1}}$) can be covered by $(2\cdot a_{\epsilon, \gamma}\cdot n + 1)\cdot 2^{2\epsilon n}$ translates of $(c_\epsilon^2 \cdot a_{\epsilon, \gamma})\cdot K\times\{0\}$.\\
	We now use the procedure from Theorem \ref{main_procedure_Q} with $R:= c_\epsilon$ and sample $N:= 2$ lattice vectors of length at most $a_{\epsilon, \gamma} \cdot R$. With probability at least $\tfrac{1}{2}\cdot 2^{-2\epsilon n}$ this succeeds and both lattice vectors are generated according to (\ref{main_procedure_Q:2a}). We condition on this event and denote by $v_1, v_2$ the resulting lattice vectors. Since $v_1$ and $v_2$ are i.i.d. and with probability at least $(2\cdot a_{\epsilon, \gamma}\cdot n + 1)^{-2}\cdot 2^{-4\epsilon n}$, there is a translate of $(c_\epsilon^2 \cdot a_{\epsilon, \gamma})\cdot K\times\{0\}$ that holds both $v_1$ and $v_2$. It follows that
	\begin{displaymath}
		v_1 - v_2 \in (2\cdot c_\epsilon^2\cdot a_{\epsilon, \gamma})\cdot K \times \{0\}.
	\end{displaymath}
	We can now use property (\ref{main_procedure_Q:2a}). With probability at least $1/4$, instead of returning $v_1$ and $v_2$, the procedure returns $v_1 + \vs$ and $v_2$. We condition on this event. Their difference is then
	\begin{displaymath}
		(v_1 + \vs)- v_2 \in (2 \cdot c_\epsilon^2\cdot a_{\epsilon, \gamma} ) \cdot K\times \{0\} + \vs
	\end{displaymath}
	We can rewrite this as
	\begin{displaymath}
		(v_1 + \vs) - v_2 =\begin{pmatrix}
			u\\
			0
		\end{pmatrix} + \vs = \begin{pmatrix}
			u-\vc\\
			0
		\end{pmatrix}
		+\begin{pmatrix}
			\tilde{t}\\
			1/n
		\end{pmatrix},
	\end{displaymath}
	for some $u \in \L$. Crucially
	\begin{displaymath}
		\|u\|_Q = \left\lVert\begin{pmatrix}
			u\\
			0
		\end{pmatrix}\right\rVert_{Q^{+1}} \leq (2\cdot c_\epsilon^2\cdot a_{\epsilon, \gamma} ).
	\end{displaymath}
	The lattice vector $\vc - u$ is then the desired approximation. Indeed, by the triangle inequality,
	\begin{displaymath}
			\|t-(\vc-u)\|_K \leq \|t-\vc\|_K + \|u\|_K \leq 1 + (2\cdot c_{\epsilon}^2 \cdot a_\epsilon) := \beta_\epsilon
	\end{displaymath}
	We set $\gamma_\epsilon := (1-1/n)^{-1}\cdot (2\beta_\epsilon + 1)$ as we have used Lemma \ref{cvp:d=n} and scaled the lattice so that $\|t-\vc\|_K \geq 1-1/n$. The lattice vector $\vc - u$ is then a $\gamma_\epsilon$-approximation to the closest lattice vector to $t$.\\
	Repeating this procedure $2^{\Omega(\epsilon) n}$ times starting from where we sampled $\tilde{t} \in t + K + c_\epsilon \cdot Q$ boosts the probability of success to $1-2^{-n}$. Using Lemma \ref{norm_plus_1}, the time and space requirements follow from those of the procedure from Theorem \ref{main_procedure_Q}.
\end{proof}

\end{document}